\documentclass[11pt]{article}
\usepackage[a4paper]{geometry}
\usepackage{amsfonts, amsmath, amssymb, amsthm, graphicx, caption, authblk, multirow, makecell, framed, float, xcolor, enumitem, tikz, hyperref}
\theoremstyle{definition}

\newtheorem{lemma}{Lemma}

\theoremstyle{remark}

\newcommand*{\mybox}[1]{%
  \framebox{\raisebox{0cm}[0.5\baselineskip][0.05\baselineskip]{%
    \hbox to 0.1cm {\hss#1\hss}}}\hspace{0.05cm}}

\begin{document}
\title{Two Standard Decks of Playing Cards are Sufficient for a ZKP for Sudoku\thanks{This paper is an extended version of \cite{sudoku2}, which appeared at COCOON 2021.}}
\author[1]{Suthee Ruangwises\thanks{\texttt{ruangwises@gmail.com}}}
\affil[1]{Department of Mathematical and Computing Science, Tokyo Institute of Technology, Tokyo, Japan}
\date{}
\maketitle

\begin{abstract}
Sudoku is a famous logic puzzle where the player has to fill a number between 1 and 9 into each empty cell of a $9 \times 9$ grid such that every number appears exactly once in each row, each column, and each $3 \times 3$ block. In 2020, Sasaki et al. developed a physical card-based protocol of zero-knowledge proof (ZKP) for Sudoku, which enables a prover to convince a verifier that he/she knows a solution of the puzzle without revealing it. Their protocol uses 90 cards, but requires nine identical copies of some cards, which cannot be found in a standard deck of playing cards (consisting of 52 different cards and two jokers). Hence, nine identical standard decks are required to perform that protocol, making the protocol not very practical. In this paper, we propose a new ZKP protocol for Sudoku that can be performed using only two standard decks of playing cards, regardless of whether the two decks are identical or different. In general, we also develop the first ZKP protocol for a generalized $n \times n$ Sudoku that can be performed using a deck of all different cards.

\end{abstract}

\section{Introduction}
\textit{Sudoku} is one of the world's most popular logic puzzles. A standard Sudoku puzzle consists of a $9 \times 9$ grid divided into nine blocks of size $3 \times 3$, with some of the cells already filled with numbers between 1 and 9. The objective of Sudoku is to fill a number into each empty cell such that every number from 1 to 9 appears exactly once in each row, each column, and each $3 \times 3$ block \cite{nikoli} (see Fig. \ref{fig1}). In a generalized version of Sudoku, the grid has size $n \times n$ and is divided into $n$ blocks of size $\sqrt{n} \times \sqrt{n}$, where $n$ is a perfect square. A generalized Sudoku is proven to be NP-complete \cite{np}.

\begin{figure}
\centering
\begin{tikzpicture}
\draw[step=0.5cm, color={rgb:black,1;white,1}] (0,0) grid (4.5,4.5);

\draw[line width=0.5mm] (0,0) -- (0,4.5);
\draw[line width=0.5mm] (1.5,0) -- (1.5,4.5);
\draw[line width=0.5mm] (3,0) -- (3,4.5);
\draw[line width=0.5mm] (4.5,0) -- (4.5,4.5);
\draw[line width=0.5mm] (0,0) -- (4.5,0);
\draw[line width=0.5mm] (0,1.5) -- (4.5,1.5);
\draw[line width=0.5mm] (0,3) -- (4.5,3);
\draw[line width=0.5mm] (0,4.5) -- (4.5,4.5);

\node at (1.75,0.25) {3};
\node at (2.25,0.25) {7};
\node at (2.75,0.25) {4};
\node at (4.25,0.25) {1};
\node at (0.75,0.75) {5};
\node at (2.75,0.75) {2};
\node at (3.25,0.75) {3};
\node at (3.75,0.75) {7};
\node at (1.25,1.25) {6};
\node at (2.75,1.25) {5};
\node at (3.25,1.25) {4};
\node at (0.25,2.25) {3};
\node at (0.75,2.25) {6};
\node at (2.25,2.25) {1};
\node at (0.25,2.75) {1};
\node at (1.25,2.75) {5};
\node at (1.75,2.75) {4};
\node at (2.75,2.75) {6};
\node at (3.75,2.75) {3};
\node at (0.25,3.25) {5};
\node at (0.75,3.25) {9};
\node at (2.25,3.25) {2};
\node at (3.25,3.25) {6};
\node at (3.75,3.25) {1};
\node at (4.25,3.25) {3};
\node at (1.25,4.25) {7};
\node at (2.75,4.25) {1};
\node at (3.75,4.25) {8};
\end{tikzpicture}
\hspace{1.2cm}
\begin{tikzpicture}
\draw[step=0.5cm, color={rgb:black,1;white,1}] (0,0) grid (4.5,4.5);

\draw[line width=0.5mm] (0,0) -- (0,4.5);
\draw[line width=0.5mm] (1.5,0) -- (1.5,4.5);
\draw[line width=0.5mm] (3,0) -- (3,4.5);
\draw[line width=0.5mm] (4.5,0) -- (4.5,4.5);
\draw[line width=0.5mm] (0,0) -- (4.5,0);
\draw[line width=0.5mm] (0,1.5) -- (4.5,1.5);
\draw[line width=0.5mm] (0,3) -- (4.5,3);
\draw[line width=0.5mm] (0,4.5) -- (4.5,4.5);

\node at (0.25,0.25) {9};
\node at (0.75,0.25) {8};
\node at (1.25,0.25) {2};
\node at (1.75,0.25) {3};
\node at (2.25,0.25) {7};
\node at (2.75,0.25) {4};
\node at (3.25,0.25) {5};
\node at (3.75,0.25) {6};
\node at (4.25,0.25) {1};
\node at (0.25,0.75) {4};
\node at (0.75,0.75) {5};
\node at (1.25,0.75) {1};
\node at (1.75,0.75) {9};
\node at (2.25,0.75) {6};
\node at (2.75,0.75) {2};
\node at (3.25,0.75) {3};
\node at (3.75,0.75) {7};
\node at (4.25,0.75) {8};
\node at (0.25,1.25) {7};
\node at (0.75,1.25) {3};
\node at (1.25,1.25) {6};
\node at (1.75,1.25) {1};
\node at (2.25,1.25) {8};
\node at (2.75,1.25) {5};
\node at (3.25,1.25) {4};
\node at (3.75,1.25) {2};
\node at (4.25,1.25) {9};
\node at (0.25,1.75) {2};
\node at (0.75,1.75) {4};
\node at (1.25,1.75) {8};
\node at (1.75,1.75) {7};
\node at (2.25,1.75) {5};
\node at (2.75,1.75) {3};
\node at (3.25,1.75) {1};
\node at (3.75,1.75) {9};
\node at (4.25,1.75) {6};
\node at (0.25,2.25) {3};
\node at (0.75,2.25) {6};
\node at (1.25,2.25) {9};
\node at (1.75,2.25) {2};
\node at (2.25,2.25) {1};
\node at (2.75,2.25) {8};
\node at (3.25,2.25) {7};
\node at (3.75,2.25) {4};
\node at (4.25,2.25) {5};
\node at (0.25,2.75) {1};
\node at (0.75,2.75) {7};
\node at (1.25,2.75) {5};
\node at (1.75,2.75) {4};
\node at (2.25,2.75) {9};
\node at (2.75,2.75) {6};
\node at (3.25,2.75) {8};
\node at (3.75,2.75) {3};
\node at (4.25,2.75) {2};
\node at (0.25,3.25) {5};
\node at (0.75,3.25) {9};
\node at (1.25,3.25) {4};
\node at (1.75,3.25) {8};
\node at (2.25,3.25) {2};
\node at (2.75,3.25) {7};
\node at (3.25,3.25) {6};
\node at (3.75,3.25) {1};
\node at (4.25,3.25) {3};
\node at (0.25,3.75) {8};
\node at (0.75,3.75) {1};
\node at (1.25,3.75) {3};
\node at (1.75,3.75) {6};
\node at (2.25,3.75) {4};
\node at (2.75,3.75) {9};
\node at (3.25,3.75) {2};
\node at (3.75,3.75) {5};
\node at (4.25,3.75) {7};
\node at (0.25,4.25) {6};
\node at (0.75,4.25) {2};
\node at (1.25,4.25) {7};
\node at (1.75,4.25) {5};
\node at (2.25,4.25) {3};
\node at (2.75,4.25) {1};
\node at (3.25,4.25) {9};
\node at (3.75,4.25) {8};
\node at (4.25,4.25) {4};
\end{tikzpicture}
\caption{An example of a $9 \times 9$ Sudoku puzzle (left) and its solution (right)}
\label{fig1}
\end{figure}

\subsection{Zero-Knowledge Proof}
We aim to construct a \textit{zero-knowledge proof (ZKP)} for Sudoku, which enables a prover $P$ to convince a verifier $V$ that he/she knows a solution of the puzzle without revealing any information about it. Formally, a ZKP is an interactive proof between $P$ and $V$ where both of them are given a computational problem $x$, but only $P$ knows a solution $w$ of $x$. A ZKP with perfect completeness and perfect soundness must satisfy the following three properties.

\begin{enumerate}
	\item \textbf{Perfect Completeness:} If $P$ knows $w$, then $V$ always accepts.
	\item \textbf{Perfect Soundness:} If $P$ does not know $w$, then $V$ always rejects.
	\item \textbf{Zero-knowledge:} $V$ learns nothing about $w$. Formally, there exists a probabilistic polynomial time algorithm $S$ (called a \textit{simulator}), not knowing $w$ but having a black-box access to $V$, such that the outputs of $S$ follow the same probability distribution as the outputs of the actual protocol.
\end{enumerate}

The concept of a ZKP was first introduced by Goldwasser et al. \cite{zkp0}. Instead of computational ZKPs, recently many results have been focusing on constructing physical ZKPs using portable objects found in everyday life such as a deck of cards. These physical protocols have benefits that they do not require electronic devices and also allow external observers to check that the prover truthfully executes the protocol (which is often a challenging task for digital protocols). They also have great didactic values and can be used to teach the concept of a ZKP to non-experts.

\section{Previous Protocols}
In 2009, Gradwohl et al. \cite{sudoku0} developed the first card-based ZKP protocols for Sudoku, and also the first for any kind of logic puzzle. Each of the six developed protocols, however, either has a nonzero soundness error or requires special tools such as scratch-off cards. Later in 2020, Sasaki et al. \cite{sudoku} developed an improved ZKP protocol for Sudoku that achieves perfect soundness without using special tools.

\subsection{Uniqueness Verification Protocol} \label{unique}
Before showing the protocol of Sasaki et al., we first explain the following \textit{uniqueness verification protocol}, which was also developed by the same authors \cite{sudoku}. This protocol allows the prover $P$ to convince the verifier $V$ that a sequence $\sigma$ of $n$ face-down cards is a permutation of different cards $a_1,a_2,...,a_n$ in some order, without revealing their orders. It also preserves the orders of the cards in $\sigma$ (so that the sequence can be later used in other protocols).

Let $x_1,x_2,...,x_n$ be another set of $n$ different cards. $P$ performs the following steps.

\begin{figure}[H]
\centering
\begin{tikzpicture}
\node at (-0.5,1.2) {$\sigma$:};
\node at (0.0,1.2) {\mybox{?}};
\node at (0.6,1.2) {\mybox{?}};
\node at (1.2,1.2) {...};
\node at (1.8,1.2) {\mybox{?}};

\node at (0.0,0.4) {\mybox{?}};
\node at (0.6,0.4) {\mybox{?}};
\node at (1.2,0.4) {...};
\node at (1.8,0.4) {\mybox{?}};

\node at (0.05,0) {$x_1$};
\node at (0.65,0) {$x_2$};
\node at (1.85,0) {$x_n$};
\end{tikzpicture}
\caption{A $2 \times n$ matrix constructed in Step 1}
\label{fig2}
\end{figure}

\begin{enumerate}
	\item Publicly place face-down cards $x_1,x_2,...,x_n$ below the face-down sequence $\sigma$ in this order from left to right to form a $2 \times n$ matrix of cards (see Fig \ref{fig2}).
	\item Rearrange all columns of the matrix by a uniformly random permutation. (This step can be performed in real world by putting both cards in each column into an envelope and scrambling all envelopes together.)
	\item Turn over all cards in the top row. $V$ verifies that the sequence is a permutation of $a_1,a_2,...,a_n$. Otherwise, $V$ rejects.
	\item Turn over all face-up cards. Rearrange all columns of the matrix by a uniformly random permutation.
	\item Turn over all cards in the bottom row. Rearrange the columns such that the cards in the bottom rows are $x_1,x_2,...,x_n$ in this order from left to right. The sequence in the top row now returns to its original state.
\end{enumerate}

\subsection{Protocol of Sasaki et al.}
Sasaki et al. \cite{sudoku} developed a protocol to verify a solution of an $n \times n$ Sudoku puzzle. This protocol has three slightly different variants. Here we will show only the first variant, which is the one using the least number of cards.

Each card used in this protocol has a positive number on the front side (denoted by \hbox{\mybox{1}}, \hbox{\mybox{2},} ...). All cards have identical back sides (denoted by \mybox{?}). First, on each cell already having a number $j$, $P$ publicly places a face-down \mybox{$j$}; on each empty cell that has a number $j$ in $P$'s solution, $P$ secretly places a face-down \mybox{$j$}.

$P$ then applies the uniqueness verification protocol to verify that every row, column, and block contains a permutation of \mybox{1}, \mybox{2}, ..., \mybox{$n$}.

In total, this protocol uses $n^2+n$ cards: $n$ identical copies of \mybox{1}, \mybox{2}, ..., \mybox{$n$} (to encode the numbers in the grid), and another set of $n$ different cards (to use in the uniqueness verification protocol). For a standard $9 \times 9$ puzzle, the protocol uses 90 cards, which is less than the number of cards in two standard decks (108). However, the protocol requires nine identical copies of \mybox{1}, \mybox{2}, ..., \mybox{9}. As a standard deck consists of 54 different cards (including two different jokers), nine identical decks are actually required in order to perform this protocol, making the protocol not very practical. Another choice is to use a different kind of deck (e.g. cards from board games) that includes several identical copies of some cards, but these decks are more difficult to find in everyday life.

Considering the drawback of this protocol, we aim to develop a more practical ZKP protocol for a $9 \times 9$ Sudoku that can be performed using only two standard decks of playing cards.

\subsection{Related Work}
After the development of card-based ZKP protocols for Sudoku, card-based ZKP protocols for other popular logic puzzles have also been proposed, including Nonogram \cite{nonogram,nonogram2}, Akari \cite{akari}, Takuzu \cite{akari,takuzu}, Kakuro \cite{akari,kakuro}, KenKen \cite{akari}, Makaro \cite{makaro}, Norinori \cite{norinori}, Slitherlink \cite{slitherlink}, Juosan \cite{takuzu}, Numberlink \cite{numberlink}, Suguru \cite{suguru}, Ripple Effect \cite{ripple}, Nurikabe \cite{nurikabe}, Hitori \cite{nurikabe}, Cryptarithmetic \cite{crypta}, and Bridges \cite{bridges}.	

Apart from verifying solutions of logic puzzles, card-based protocols have also been extensively studied in secure multi-party computation, a setting where multiple parties want to jointly compute a function of their secret inputs without revealing them. The vast majority of the developed protocols, however, also uses multiple identical copies of two different cards (usually denoted by \mybox{$\clubsuit$} and \mybox{$\heartsuit$}), making them not implementable by a single standard deck of playing cards. The only exceptions are \cite{standard3,standard4,standard2,standard1} which introduced AND, XOR, and copy protocols using a standard deck, and \cite{standardyao} which introduced a Yao's millionaire protocol using a standard deck. In \cite{standard4}, the authors also posed an open problem to develop ZKP protocols for logic puzzles using a standard deck.

Pratically, a standard deck of playing cards consists of 54 different cards (including two different jokers). Theoretically, it is also a challenging problem to develop a protocol that can be implemented using a deck of all different cards, so we also study the setting where the deck consists of \mybox{1}, \mybox{2}, ... where each card can have an arbitrarily large number on it.

\section{Our Contribution}
In this paper, we propose a new ZKP protocol for a generalized $n \times n$ Sudoku puzzle with perfect completeness and soundness using a deck of all different cards.

There are two slightly different methods to implement our protocol. The first one uses $n^2+n\sqrt{n}+n+\sqrt{n}$ cards and $4n\sqrt{n}$ shuffles. The second one uses $n^2+2n+3\sqrt{n}$ cards and at most $2n^2(\sqrt{n}-1)+2$ shuffles (see Table \ref{table1}).

In particular, for a standard $9 \times 9$ Sudoku puzzle, our protocol (with the second method of implementation) uses 108 cards and can be performed using two standard decks of playing cards, regardless of whether the two decks are identical or different (see Table \ref{table2}).

Theoretically, this work is an important step in card-based cryptography as it is the first ZKP protocol for any kind of logic puzzle that can be performed using a deck of all different cards, answering the open problem posed in \cite{standard4}.

The main difference from the conference version of this paper \cite{sudoku2} is the inclusion of an optimization of the number of shuffles in Section \ref{opt}, which was omitted in the conference version.

\begin{table}[H]
	\centering
	\begin{tabular}{|c|c|c|c|}
		\hline
		\textbf{Protocol} & \textbf{\thead{Standard\\ Deck?}} & \textbf{\#Cards} & \textbf{\#Shuffles} \\ \hline
		\textbf{Sasaki et al. \cite{sudoku}} & no & $n^2+n$ & $5n$ \\ \hline
		\textbf{Ours (\S \ref{m1})} & yes & $n^2+n\sqrt{n}+n+\sqrt{n}$ & $4n\sqrt{n}$ \\ \hline
		\textbf{Ours (\S \ref{m2})} & yes & $n^2+2n+3\sqrt{n}$ & \thead{$2n^2(\sqrt{n}-1)$ for even $n$\\ $2n^2(\sqrt{n}-1)+2$ for odd $n>9$} \\ \hline
	\end{tabular}
	\medskip
	\caption{The number of required cards and shuffles for each protocol for an $n \times n$ Sudoku} \label{table1}
\end{table}

\begin{table}[H]
	\centering
	\begin{tabular}{|c|c|c|c|}
		\hline
		\textbf{Protocol} & \textbf{\thead{Standard\\ Deck?}} & \textbf{\#Cards} & \textbf{\#Shuffles} \\ \hline
		\textbf{Sasaki et al. \cite{sudoku}} & no & 90 & 45 \\ \hline
		\textbf{Ours (\S \ref{m1})} & yes & 120 & 108 \\ \hline
		\textbf{Ours (\S \ref{m2})} & yes & 108 & 322 \\ \hline
	\end{tabular}
	\medskip
	\caption{The number of required cards and shuffles for each protocol for a $9 \times 9$ Sudoku} \label{table2}
\end{table}

\section{Preliminaries}
At first, we assume that all cards used in our protocols have different front sides and identical back sides (although we will later show that some pairs of cards can have identical front sides or different back sides, and our protocols still work correctly).

\subsection{Marked Matrix}
Suppose we have a $k \times \ell$ matrix of face-down cards (we call these cards \textit{encoding cards}). Let Row $i$ denote an $i$-th topmost row and let Column $j$ denote a $j$-th leftmost column. To the left of Column 1, publicly place face-down cards $p_1,p_2,...,p_k$ in this order from top to bottom; this new column is called Column 0. Analogously, above Row 1, publicly place face-down cards $q_1,q_2,...,q_\ell$ in this order from left to right; this new row is called Row 0.

We call this new structure a $k \times \ell$ \textit{marked matrix} (see Fig \ref{fig3}), and we call the cards in Row 0 and Column 0 \textit{marking cards}.

\begin{figure}[H]
\centering
\begin{tikzpicture}
\node at (0.0,0.6) {\mybox{?}};
\node at (0.5,0.6) {\mybox{?}};
\node at (1.0,0.6) {\mybox{?}};
\node at (1.5,0.6) {\mybox{?}};
\node at (2.0,0.6) {\mybox{?}};

\node at (0.0,1.2) {\mybox{?}};
\node at (0.5,1.2) {\mybox{?}};
\node at (1.0,1.2) {\mybox{?}};
\node at (1.5,1.2) {\mybox{?}};
\node at (2.0,1.2) {\mybox{?}};

\node at (0.0,1.8) {\mybox{?}};
\node at (0.5,1.8) {\mybox{?}};
\node at (1.0,1.8) {\mybox{?}};
\node at (1.5,1.8) {\mybox{?}};
\node at (2.0,1.8) {\mybox{?}};

\node at (0.0,2.4) {\mybox{?}};
\node at (0.5,2.4) {\mybox{?}};
\node at (1.0,2.4) {\mybox{?}};
\node at (1.5,2.4) {\mybox{?}};
\node at (2.0,2.4) {\mybox{?}};

\node at (0.0,3.2) {\mybox{?}};
\node at (0.5,3.2) {\mybox{?}};
\node at (1.0,3.2) {\mybox{?}};
\node at (1.5,3.2) {\mybox{?}};
\node at (2.0,3.2) {\mybox{?}};

\node at (0.0,3.6) {$q_1$};
\node at (0.5,3.6) {$q_2$};
\node at (1.0,3.6) {$q_3$};
\node at (1.5,3.6) {$q_4$};
\node at (2.0,3.6) {$q_5$};

\node at (-0.7,0.6) {\mybox{?}};
\node at (-0.7,1.2) {\mybox{?}};
\node at (-0.7,1.8) {\mybox{?}};
\node at (-0.7,2.4) {\mybox{?}};

\node at (-1.1,0.6) {$p_4$};
\node at (-1.1,1.2) {$p_3$};
\node at (-1.1,1.8) {$p_2$};
\node at (-1.1,2.4) {$p_1$};

\draw[] (-1.5,0.3) -- (-1.5,4.7);
\draw[] (-2.5,3.9) -- (2.7,3.9);

\node at (-1.8,0.6) {4};
\node at (-1.8,1.2) {3};
\node at (-1.8,1.8) {2};
\node at (-1.8,2.4) {1};
\node at (-1.8,3.2) {0};
\node at (-2.5,1.9) {Row};

\node at (-0.7,4.2) {0};
\node at (0.0,4.2) {1};
\node at (0.5,4.2) {2};
\node at (1.0,4.2) {3};
\node at (1.5,4.2) {4};
\node at (2.0,4.2) {5};
\node at (0.65,4.7) {Column};
\end{tikzpicture}
\caption{An example of a $4 \times 5$ marked matrix}
\label{fig3}
\end{figure}

\subsection{Shuffle Operations}
Given a $k \times \ell$ marked matrix and a set $S \subseteq \{1,2,...,k\}$, an operation \texttt{row\_shuffle}($S$) rearranges the rows of the matrix with indices in $S$ (including marking cards in Column 0) by a uniformly random permutation. For example, \texttt{row\_shuffle}($\{3,4,5\}$) rearranges Row 3, Row 4, and Row 5 of the matrix by a uniformly random permutation. This operation can be performed in real world by putting all cards in each row with an index in $S$ into an envelope and scrambling all envelopes together.

Analogously, for a set $S \subseteq \{1,2,...,\ell\}$, an operation \texttt{col\_shuffle}($S$) rearranges the columns of the matrix with indices in $S$ (including marking cards in Row 0) by a uniformly random permutation.

\subsection{Rearrangement Protocol} \label{revert}
After applying some shuffle operations to a marked matrix, a \textit{rearrangement protocol} reverts the matrix back to its original state. Slightly different variants of this protocol with the same idea has been used in previous work \cite{makaro,revert1,revert2,numberlink,ripple,sudoku}.

Suppose we have a $k \times \ell$ marked matrix $M$ with marking cards $p_1,p_2,...,p_k$ in Column 0 and $q_1,q_2,...,q_\ell$ in Row 0. We perform the following steps.

\begin{enumerate}
	\item Apply \texttt{row\_shuffle}($\{1,2,...,k\}$) and \texttt{col\_shuffle}($\{1,2,...,\ell\}$) to $M$.
	\item Turn over all marking cards in Column 0 and Row 0. Rearrange the rows of $M$ such that the marking cards in Column 0 are $p_1,p_2,...,p_k$ in this order from top to bottom. Rearrange the columns of $M$ such that the marking cards in Row 0 are $q_1,q_2,...,q_\ell$ in this order from left to right.
\end{enumerate}

\subsection{Standard Deck Chosen Cut Protocol} \label{chosen}
Given a $k \times \ell$ marked matrix $M$, a \textit{standard deck chosen cut protocol} allows the prover $P$ to choose a card located at Row $i$ and Column $j$ of $M$ he/she wants without revealing $i$ or $j$ to the verifier $V$. This protocol was modified from an original \textit{chosen cut protocol} of Koch and Walzer \cite{koch} (which uses identical copies of a \mybox{$\clubsuit$} and a \mybox{$\heartsuit$}) so that it can be performed using a standard deck. $P$ performs the following steps.

\begin{enumerate}
	\item On each of the $k\ell$ encoding cards in the matrix, secretly stack each of face-down cards $x_1,x_2,...,x_{k\ell}$ (called \textit{helping cards}) such that $x_1$ is located at Row $i$ and Column $j$, and $x_2,x_3,...,x_{k\ell}$ are in a uniformly random permutation (which is known to $P$ but not to $V$).
	\item Apply \texttt{row\_shuffle}($\{1,2,...,k\}$) and \texttt{col\_shuffle}($\{1,2,...,\ell\}$) to $M$.
	\item Turn over all helping cards. Locate the position of $x_1$. The encoding card from that stack is the one originally located at Row $i$ and Column $j$ as desired.
	\item Remove all helping cards. Apply the rearrangement protocol to revert $M$ to its original state.
\end{enumerate}

This protocol will be implicitly used in our main protocol, with Step 2 being replaced by equivalent operations.

\section{Main Protocol}
For simplicity, we will show a protocol for a standard $9 \times 9$ Sudoku puzzle. Our protocol can be straightforwardly generalized to an $n \times n$ puzzle.

We use the following cards in our protocol.
\begin{itemize}
	\item encoding cards $a_j,b_j,c_j,d_j,e_j,f_j,g_j,h_j,i_j$ ($j=1,2,...,9$)
	\item marking cards $p_j$ ($j=1,2,3$) and $q_j$ ($j=1,2,...,9$)
	\item helping cards $x_j,y_j,z_j$ ($j=1,2,...,9$)
\end{itemize}

Suppose the grid is divided into blocks $A, B, ..., I$ (see Fig. \ref{figM1}). We use a card $a_j$ ($j=1,2,...,9$) to encode a number $j$ in Block $A$. Analogously, we use cards $b_j,c_j,...,i_j$ ($j=1,2,...,9$) to encode numbers $j$ in blocks $B, C, ..., I$, respectively.

\begin{figure}[H]
\centering
\begin{tikzpicture}
\draw[step=0.5cm, color={rgb:black,1;white,1}] (0,0) grid (4.5,4.5);

\draw[line width=0.5mm] (0,0) -- (0,4.5);
\draw[line width=0.5mm] (1.5,0) -- (1.5,4.5);
\draw[line width=0.5mm] (3,0) -- (3,4.5);
\draw[line width=0.5mm] (4.5,0) -- (4.5,4.5);
\draw[line width=0.5mm] (0,0) -- (4.5,0);
\draw[line width=0.5mm] (0,1.5) -- (4.5,1.5);
\draw[line width=0.5mm] (0,3) -- (4.5,3);
\draw[line width=0.5mm] (0,4.5) -- (4.5,4.5);

\node at (0.75,3.75) {\resizebox{!}{8mm}{\color{gray} A}};
\node at (2.25,3.75) {\resizebox{!}{8mm}{\color{gray} B}};
\node at (3.75,3.75) {\resizebox{!}{8mm}{\color{gray} C}};
\node at (0.75,2.25) {\resizebox{!}{8mm}{\color{gray} D}};
\node at (2.25,2.25) {\resizebox{!}{8mm}{\color{gray} E}};
\node at (3.75,2.25) {\resizebox{!}{8mm}{\color{gray} F}};
\node at (0.75,0.75) {\resizebox{!}{8mm}{\color{gray} G}};
\node at (2.25,0.75) {\resizebox{!}{8mm}{\color{gray} H}};
\node at (3.75,0.75) {\resizebox{!}{8mm}{\color{gray} I}};
\end{tikzpicture}
\caption{Blocks $A, B, C, D, E, F, G, H$, and $I$ in the grid}
\label{figM1}
\end{figure}
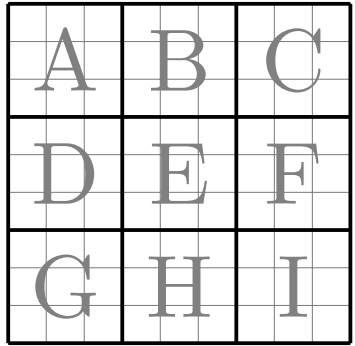

On each cell already having a number, $P$ publicly places a face-down corresponding card (e.g. places a card $b_3$ on a cell with a number 3 in Block $B$). On each empty cell, $P$ secretly places a face-down corresponding card according to his/her solution.

\subsection{Block Verification} \label{m0}
First, $P$ performs the following steps to verify that every number from 1 to 9 appears exactly once in each block.

\begin{enumerate}
	\item Apply the uniqueness verification protocol in Section \ref{unique} to verify that Block $A$ consists of cards $a_1,a_2,...,a_9$ in some order.
	\item Analogously perform Step 1 for Blocks $B, C, ..., I$.
\end{enumerate}

Now $V$ is convinced that every number from 1 to 9 appears exactly once in each block.

\subsection{Row/Column Verification}
Next, $P$ will verify that every number from 1 to 9 appears exactly once in each row and column. There are two methods to do this.

\subsubsection{Method A} \label{m1}
$P$ performs Steps 1 to 6 as shown below to verify that a number 1 appears exactly once in each of the three topmost rows.

\begin{enumerate}
	\item Take the cards from the three topmost rows to form a $3 \times 9$ matrix and publicly place marking cards $p_1,p_2,p_3$ in Column 0 and $q_1,q_2,...,q_9$ in Row 0 to create a $3 \times 9$ marked matrix $M$.
	\item On each encoding card in Block $A$, secretly stack each of face-down cards $x_1,x_2,...,x_9$ such that $x_1$ is on $a_1$, and $x_2,x_3,...,x_9$ are in a uniformly random permutation (which is known to $P$ but not to $V$).
	\item Do the same for cards $y_1,y_2,...,y_9$ in Block $B$ (with $y_1$ on $b_1$) and $z_1,z_2,...,z_9$ in Block $C$ (with $z_1$ on $c_1$).
	\item Apply \texttt{row\_shuffle}($\{1,2,3\}$), \texttt{col\_shuffle}($\{1,2,3\}$), \texttt{col\_shuffle}($\{4,5,6\}$), and \linebreak \texttt{col\_shuffle}($\{7,8,9\}$) to $M$.
	\item Turn over all helping cards. Locate the positions of $x_1$, $y_1$, and $z_1$. Turn over the encoding cards in these three stacks to show that they are $a_1$, $b_1$, and $c_1$, respectively, and that they are all located at different rows. Otherwise, $V$ rejects.
	\item Remove all helping cards and turn all encoding cards face-down. Apply the rearrangement protocol in Section \ref{revert} to revert $M$ to its original state.
\end{enumerate}

Note that Steps 2 to 6 are equivalent to applying the standard deck chosen cut protocol in Section \ref{chosen} to Blocks $A$, $B$, and $C$, simultaneously. These steps ensure that the three 1s in Blocks $A$, $B$, and $C$ are all located at different rows. Since it has already been shown that each block contains exactly one 1, this implies there is exactly one 1 in each of the three topmost rows.

\begin{enumerate}
	\setcounter{enumi}{6}
	\item Perform Steps 1 to 6 analogously for numbers $2, 3, ..., 9$. Now $V$ is convinced that every number appears exactly once in each of the three topmost rows.
	\item Perform Steps 1 to 7 analogously for Blocks $D$, $E$, and $F$, and for Blocks $G$, $H$, and $I$ to verify the rest of the rows. The verification for columns also works similarly (take the cards from Blocks $A$, $D$, and $G$, from Blocks $B$, $E$, and $H$, and from Blocks $C$, $F$, and $I$, and just transpose the marked matrix).
\end{enumerate}

Now $V$ is convinced that every number from 1 to 9 appears exactly once in each block, each row, and each column.

This method uses 81 encoding cards, 12 marking cards, and 27 helping cards, resulting in the total of 120 cards, slightly more than the number of cards in two standard decks. It uses $18 + 6 \times 9 \times 3 \times 2 = 342$ shuffles (which can be reduced to 108 after the optimization in Section \ref{opta}). We aim to further reduce the number of required cards as a trade-off between the numbers of cards and shuffles.

\subsubsection{Method B} \label{m2}
Observe that in Steps 1 to 6 of Method A, we verify that the three 1s in Blocks $A$, $B$, and $C$ are all located at different rows by verifying these three blocks at the same time, which requires a lot of marking and helping cards. Instead, we can first verify that the two 1s in Blocks $A$ and $B$ are located at different rows, then do the same for Blocks $A$ and $C$, and for Blocks $B$ and $C$. This leads to the same conclusion that the three 1s in Blocks $A$, $B$, and $C$ are all located at different rows.

$P$ performs Steps 1 to 6 as shown below to verify that the two 1s in Blocks $A$ and $B$ are located at different rows.

\begin{enumerate}
	\item Take the cards from blocks $A$ and $B$ to form a $3 \times 6$ matrix and publicly place marking cards $p_1,p_2,p_3$ in Column 0 and $q_1,q_2,...,q_6$ in Row 0 to create a $3 \times 6$ marked matrix $M$.
	\item On each encoding card in Block $A$, secretly stack each of face-down cards $x_1,x_2,...,x_9$ such that $x_1$ is on $a_1$, and $x_2,x_3,...,x_9$ are in a uniformly random permutation (which is known to $P$ but not to $V$).
	\item Do the same for cards $y_1,y_2,...,y_9$ in Block $B$ (with $y_1$ on $b_1$).
	\item Apply \texttt{row\_shuffle}($\{1,2,3\}$), \texttt{col\_shuffle}($\{1,2,3\}$), and \texttt{col\_shuffle}($\{4,5,6\}$) to $M$.
	\item Turn over all helping cards. Locate the positions of $x_1$ and $y_1$. Turn over the encoding cards in both stacks to show that they are $a_1$ and $b_1$, respectively, and that they are located at different rows. Otherwise, $V$ rejects.
	\item Remove all helping cards and turn all encoding cards face-down. Apply the rearrangement protocol in Section \ref{revert} to revert $M$ to its original state.
\end{enumerate}

Now $V$ is convinced that the two 1s in Blocks $A$ and $B$ are located at different rows.

\begin{enumerate}
	\setcounter{enumi}{6}
	\item Perform Steps 1 to 6 analogously for numbers $2, 3, ..., 9$.
	\item Perform Steps 1 to 7 analogously for Blocks $A$ and $C$, and for Blocks $B$ and $C$. Now $V$ is convinced that every number appears exactly once in each of the three topmost rows.
	\item Perform Steps 1 to 8 analogously to verify the rest of the rows. The verification for columns also works similarly.
\end{enumerate}

Now $V$ is convinced that every number from 1 to 9 appears exactly once in each block, each row, and each column.

This method uses 81 encoding cards, nine marking cards, and 18 helping cards, resulting in the total of 108 cards, which is exactly the number of cards from two standard decks (including jokers). It uses $18 + 5 \times 9 \times 3 \times 3 \times 2 = 828$ shuffles (which can be reduced to 322 after the optimization in Section \ref{optb}).

We say that two cards are from the same set if they are denoted by the same letter with different indices (e.g. $d_2$ and $d_5$ are from the same set). Notice that in both methods, cards from different sets never get mixed together. Therefore, cards from different sets can have identical front sides or different back sides (or even different sizes) and our protocol still works correctly. The only requirement is that all cards from the same set must have different front sides and identical back sides.

Therefore, when implementing Method B using two standard decks of playing cards, we can, for example, use 54 cards from the first deck in the sets $a_j,b_j,...,f_j$ and 54 cards from the second deck in the remaining sets. The protocol always works correctly regardless of whether the two decks are identical or different, since it allows cards from different sets to have identical front sides (in case of identical decks) or different back sides or sizes (in case of different decks). Note that in some decks, the two jokers are identical; in that case, we just need to make sure that the two jokers are in different sets.

\subsection{Generalization}
This protocol can be straightforwardly generalized to an $n \times n$ Sudoku puzzle.

Method A uses $n^2$ encoding cards, $n+\sqrt{n}$ marking cards, and $n\sqrt{n}$ helping cards, resulting in the total of $n^2+n\sqrt{n}+n+\sqrt{n}$ cards. It uses $2n + (\sqrt{n}+3) \times n \times \sqrt{n} \times 2 = 2n^2+6n\sqrt{n}+2n$ shuffles (which can be reduced to $4n\sqrt{n}$ after the optimization in Section \ref{opta}).

Method B uses $n^2$ encoding cards, $3\sqrt{n}$ marking cards, and $2n$ helping cards, resulting in the total of $n^2+2n+3\sqrt{n}$ cards. It uses $2n + 5 \times n \times \binom{\sqrt{n}}{2} \times \sqrt{n} \times 2 = 5n^2(\sqrt{n}-1)+2n$ shuffles (which can be reduced to at most $2n^2(\sqrt{n}-1)+2$ after the optimization in Section \ref{optb}).

\section{Proof of Correctness and Security}
We will prove the perfect completeness, perfect soundness, and zero-knowledge properties of our protocol.

\begin{lemma}[Perfect Completeness] \label{lem1}
If $P$ knows a solution of the Sudoku puzzle, then $V$ always accepts.
\end{lemma}

\begin{proof}
Suppose $P$ knows a solution and places cards on the grid accordingly. Every number from 1 to 9 will appear exactly once in each row, each column, and each block. Hence, the uniqueness verification protocol will pass for every block. Also, the same numbers from different blocks are always located at different rows and columns, so both Methods A and B will pass. Therefore, $V$ always accepts.
\end{proof}

\begin{lemma}[Perfect Soundness] \label{lem2}
If $P$ does not know a solution of the Sudoku puzzle, then $V$ always rejects.
\end{lemma}

\begin{proof}
Suppose $P$ does not know a solution. There will be a number that appears at least twice in the same row, column, or block. If it appears twice in a block, the uniqueness verification protocol for that block will fail. If it appears twice in different blocks in the same row (resp. column), Method A will fail when verifying the three blocks containing that row (resp. column); also, method B will fail when verifying the two blocks where these two numbers appear. Therefore, $V$ always rejects.
\end{proof}

\begin{lemma}[Zero-Knowledge] \label{lem3}
During the verification, $V$ learns nothing about $P$'s solution.
\end{lemma}

\begin{proof}
It is sufficient to show that all distributions of cards that are turned face-up can be simulated by a simulator $S$ that does not know $P$'s solution.

\begin{itemize}
	\item In Steps 3 and 5 of the uniqueness verification protocol in Section \ref{unique}, the orders of the $n$ cards are uniformly distributed among all $n!$ permutations. Hence, it can be simulated by $S$.
	\item In Step 2 of the rearrangement protocol in Section \ref{revert}, the orders of $p_1,p_2,...,$ $p_k$ and $q_1,q_2,...,q_\ell$ are uniformly distributed among all $k!$ permutations and $\ell!$ permutations, respectively. Hence, it can be simulated by $S$.
	\item In Step 5 of Method A in Section \ref{m1}, the rows where $x_1$, $y_1$, and $z_1$ are located are uniformly distributed among all $3!=6$ permutations of the first three rows; the columns where they are located are uniformly distributed among all $3^3=27$ combinations of three columns from Blocks $A$, $B$, and $C$. Also, the orders of $x_2,x_3,...,x_9$ are uniformly distributed among all $8!$ permutations of the remaining cards in Block $A$; the same goes for $y_2,y_3,...,y_9$ in Block $B$ and $z_2,z_3,...,z_9$ in Block $C$. Hence, it can be simulated by $S$.
	\item In Step 5 of Method B in Section \ref{m2}, the rows where $x_1$ and $y_1$ are located are uniformly distributed among all $\frac{3!}{1!}=6$ permutations of two rows chosen from the first three rows; the columns where they are located are uniformly distributed among all $3^2=9$ combinations of two columns from Blocks $A$ and $B$. Also, the orders of $x_2,x_3,...,x_9$ are uniformly distributed among all $8!$ permutations of the remaining cards in Block $A$; the same goes for $y_2,y_3,...,y_9$ in Block $B$. Hence, it can be simulated by $S$.
\end{itemize}
\end{proof}

\section{Optimization of the Number of Shuffles} \label{opt}
\subsection{Method A} \label{opta}
\begin{itemize}
	\item In the block verification in Section \ref{m0}, we can verify three blocks at a time using cards $x_j,y_j,z_j$ ($j=1,2,...,9$) as 27 additional cards in the uniqueness verification protocol (with a condition that all cards in the sets $x_j$, $y_j$, and $z_j$ must have different front sides and identical back sides). This reduces the number of shuffles by $2 \times 9 - 2 \times 3 = 12$ from 342 to 330.
	
	\item In Step 7 of Method A in Section \ref{m1}, we do not need to verify that a number 9 appears exactly once in each row and column. Since we have already verified that each of the numbers $1,2,...,8$ appears exactly once in each row (resp. column), the only remaining position in each row (resp. column) must contain a 9. This reduces the number of shuffles by $6 \times 3 \times 2 = 36$ from 330 to 294.
	
	\item In Step 4 of Method A, we can apply \texttt{col\_shuffle}($\{1,2,...,9\}$) instead of \texttt{col\_shuffle}\linebreak($\{1,2,3\}$), \texttt{col\_shuffle}($\{4,5,6\}$), and \texttt{col\_shuffle}($\{7,8,9\}$) to $M$ (with a condition that all cards in the sets $x_j$, $y_j$, and $z_j$ must have different front sides, so that we can tell different blocks apart after turning over helping cards). This reduces the number of shuffles by $2 \times 8 \times 3 \times 2 = 96$ from 294 to 198.
	
	\item In Step 6 of Method A, After verifying that numbers $j$ ($j=1,2,...,7$) in three selected blocks are all located at different rows (resp. columns), we do not have to revert $M$ back to its original state. Since $P$ knows exactly where the number $j+1$ in each block is (because $P$ knows which helping card is stacked on the encoding card corresponding to number $j+1$ from Steps 2 and 3), $P$ can immediately start the next round by performing the chosen cut protocol to find the number $j+1$. This reduces the number of shuffles by $2 \times 7 \times 3 \times 2 = 84$ from 198 to 114.
	
	\item In Step 8 of Method A, after verifying that numbers 8 in three selected blocks are all located at different columns, we do not have to revert $M$ back to its original state since the cards in these three blocks will not be used anymore. This reduces the number of shuffles by $2 \times 3 = 6$ from 114 to 108.
\end{itemize}

The formal steps of the optimized protocol for row/column verification are as follows.

\begin{enumerate}
	\item Take the cards from the three topmost rows to form a $3 \times 9$ matrix and publicly place marking cards $p_1,p_2,p_3$ in Column 0 and $q_1,q_2,...,q_9$ in Row 0 to create a $3 \times 9$ marked matrix $M$.
	\item On each encoding card in Block $A$, secretly stack each of face-down cards $x_1,x_2,...,x_9$ such that $x_1$ is on $a_1$, and $x_2,x_3,...,x_9$ are in a uniformly random permutation (which is known to $P$ but not to $V$).
	\item Do the same for cards $y_1,y_2,...,y_9$ in Block $B$ (with $y_1$ on $b_1$) and $z_1,z_2,...,z_9$ in Block $C$ (with $z_1$ on $c_1$).
	\item Apply \texttt{row\_shuffle}($\{1,2,3\}$), \texttt{col\_shuffle}($\{1,2,...,9\}$) to $M$.
	\item Turn over all helping cards. Locate the positions of $x_1$, $y_1$, and $z_1$. Turn over the encoding cards in these three stacks to show that they are $a_1$, $b_1$, and $c_1$, respectively, and that they are all located at different rows. Otherwise, $V$ rejects.
	\item Remove all helping cards and turn all encoding cards face-down.
	\item Perform Steps 1 to 6 analogously for numbers $2, 3, ..., 8$. Apply the rearrangement protocol in Section \ref{revert} to revert $M$ to its original state.
	\item Perform Steps 1 to 7 analogously for Blocks $D$, $E$, and $F$, and for Blocks $G$, $H$, and $I$ to verify the rest of the rows. The verification for columns also works similarly, but without applying the rearrangement protocol in Step 7.
\end{enumerate}

For an $n \times n$ puzzle, originally this method uses $2n + (\sqrt{n}+3) \times n \times \sqrt{n} \times 2 = 2n^2+6n\sqrt{n}+2n$ shuffles. After the optimization, it uses $2\sqrt{n} + (2(n-1)+2) \times \sqrt{n} \times 2 - 2\sqrt{n} = 4n\sqrt{n}$ shuffles.

\subsection{Method B} \label{optb}
\begin{itemize}
	\item In the block verification in Section \ref{m0}, we can verify three blocks at a time using cards $x_j,y_j$ ($j=1,2,...,9$), $p_j$ ($j=1,2,3$), and $q_j$ ($j=1,2,...,6$) as 27 additional cards in the uniqueness verification protocol (with a condition that all cards in the sets $x_j$, $y_j$, $p_j$ and $q_j$ must have different front sides and identical back sides). This reduces the number of shuffles by $2 \times 9 - 2 \times 3 = 12$ from 828 to 816.
\end{itemize}
	
	Note that this optimization cannot be straightforwardly generalized to an $n \times n$ puzzle. For an $n \times n$ puzzle with $n>9$, we can verify two blocks (not $\sqrt{n}$ blocks) at a time using $2n$ helping cards (as we have only $2n+3\sqrt{n}$ marking and helping cards). Hence, the block verification uses $2\lceil \frac{n}{2} \rceil$ shuffles.
	
\begin{itemize}	
	\item In Step 7 of Method B in Section \ref{m2}, we do not need to verify that a number 9 appears exactly once in each row and column. Since we have already verified that each of the numbers $1,2,...,8$ appears exactly once in each row (resp. column), the only remaining position in each row (resp. column) must contain a 9. This reduces the number of shuffles by $5 \times 3 \times 3 \times 2 = 90$ from 816 to 726.
	
	\item In Step 4 of Method B, we can apply \texttt{col\_shuffle}($\{1,2,...,6\}$) instead of \texttt{col\_shuffle}\linebreak($\{1,2,3\}$) and \texttt{col\_shuffle}($\{4,5,6\}$) to $M$ (with a condition that all cards in the sets $x_j$ and $y_j$ must have different front sides, so that we can tell different blocks apart after turning over helping cards). This reduces the number of shuffles by $8 \times 3 \times 3 \times 2 = 144$ from 726 to 582.
	
	\item In Step 6 of Method B, after verifying that numbers $j$ ($j=1,2,...,7$) in two selected blocks are located at different rows (resp. columns), we do not have to revert $M$ back to its original state. Since $P$ knows exactly where the number $j+1$ in each block is (because $P$ knows which helping card is stacked on the encoding card corresponding to number $j+1$ from Steps 2 and 3), $P$ can immediately start the next round by performing the chosen cut protocol to find the number $j+1$. This reduces the number of shuffles by $2 \times 7 \times 3 \times 3 \times 2 = 252$ from 582 to 330.
	
	\item Among the $3 \times 3 \times 2 = 18$ adjacent pairs of blocks we have to verify in Steps 8 and 9 of Method B, notice that the order of pairs we verify does not matter. Hence, we can set the order of verification such that the last four pairs of blocks we verify are Blocks $A$ and $D$, Blocks $B$ and $E$, Blocks $C$ and $F$, and Blocks $G$ and $H$. For these four pairs of blocks, after verifying that numbers 8 in the two blocks are located at different rows or columns, we do not have to revert $M$ back to its original state since the cards in these two blocks will not be used anymore. This reduces the number of shuffles by $2 \times 4 = 8$ from 330 to 322.
\end{itemize}
	
	Note that for an $n \times n$ puzzle, we set the order of verification such that the last $\lfloor \frac{n}{2}\rfloor$ pairs of blocks we verify contain $2\lfloor \frac{n}{2}\rfloor$ different blocks, hence reducing the number of shuffles by $2\lfloor \frac{n}{2}\rfloor$.

The formal steps of the optimized protocol for row/column verification are as follows.

\begin{enumerate}
	\item Take the cards from blocks $A$ and $B$ to form a $3 \times 6$ matrix and publicly place marking cards $p_1,p_2,p_3$ in Column 0 and $q_1,q_2,...,q_6$ in Row 0 to create a $3 \times 6$ marked matrix $M$.
	\item On each encoding card in Block $A$, secretly stack each of face-down cards $x_1,x_2,...,x_9$ such that $x_1$ is on $a_1$, and $x_2,x_3,...,x_9$ are in a uniformly random permutation (which is known to $P$ but not to $V$).
	\item Do the same for cards $y_1,y_2,...,y_9$ in Block $B$ (with $y_1$ on $b_1$).
	\item Apply \texttt{row\_shuffle}($\{1,2,3\}$), \texttt{col\_shuffle}($\{1,2,...,6\}$) to $M$.
	\item Turn over all helping cards. Locate the positions of $x_1$ and $y_1$. Turn over the encoding cards in both stacks to show that they are $a_1$ and $b_1$, respectively, and that they are located at different rows. Otherwise, $V$ rejects.
	\item Remove all helping cards and turn all encoding cards face-down.
	\item Perform Steps 1 to 6 analogously for numbers $2, 3, ..., 8$. Apply the rearrangement protocol in Section \ref{revert} to revert $M$ to its original state.
	\item Perform Steps 1 to 7 analogously for other 17 pairs of adjacent blocks in any order, but the last four pairs must be Blocks $A$ and $D$, Blocks $B$ and $E$, Blocks $C$ and $F$, and Blocks $G$ and $H$. Do not apply the rearrangement protocol in Step 7 for the last four pairs.
\end{enumerate}

For an $n \times n$ puzzle with $n>9$, originally this method uses $2n + 5 \times n \times \binom{\sqrt{n}}{2} \times \sqrt{n} \times 2 = 5n^2(\sqrt{n}-1)+2n$ shuffles. After the optimization, it uses $2\lceil \frac{n}{2} \rceil + (2(n-1)+2) \times \binom{\sqrt{n}}{2} \times \sqrt{n} \times 2 - 2\lfloor \frac{n}{2}\rfloor = 2n^2(\sqrt{n}-1)$ shuffles for an even $n$ and $2n^2(\sqrt{n}-1)+2$ shuffles for an odd $n$.

\section{Future Work}
We developed the first card-based ZKP protocol for Sudoku, and also the first one for any kind of logic puzzle, that can be performed using a deck of all different cards. Our protocol for a standard $9 \times 9$ Sudoku can be performed using two standard decks of playing cards, regardless of whether the two decks are identical or different. However, the drawback of our protocol is that it uses a large number of shuffles, which makes it not very practical. A possible future work is to develop an equivalent protocol for Sudoku that uses asymptotically less number of shuffles. Other challenging future work includes developing ZKP protocols for other logic puzzles (e.g. Kakuro, Numberlink) that can be performed using a deck of all different cards.

\end{document}